\documentclass[final]{IEEEtran}
\usepackage[latin9]{inputenc}
\usepackage{array}
\usepackage{float}
\usepackage{bm}
\usepackage{amsmath}
\usepackage{amssymb}
\usepackage{graphicx}
\usepackage{amsthm}
\usepackage{epstopdf}
\makeatletter

\providecommand{\tabularnewline}{\\}
\floatstyle{ruled}
\newfloat{algorithm}{tbp}{loa}
\providecommand{\algorithmname}{Algorithm}
\floatname{algorithm}{\protect\algorithmname}

\newcommand{\MYfooter}{\smash{
\hfil\parbox[t][\height][t]{\textwidth}{}\hfil\hbox{}}}

\def\ps@IEEEtitlepagestyle{%
\def\@oddfoot{\MYfooter}%
\def\@evenfoot{\MYfooter}}
\newtheorem{theorem}{Theorem}

\newcommand{\inner}[2]{{\langle#1,#2\rangle}}

\makeatother

\begin{document}
\global\long\def\bd{{\bf d}}

\global\long\def\real{\mathbb{R}}

\global\long\def\bx{{\bf x}}

\global\long\def\bb{{\bf b}}

\global\long\def\bg{{\bf g}}

\global\long\def\bd{{\bf d}}

\global\long\def\bo{{\bf 0}}

\global\long\def\by{{\bf y}}
 \global\long\def\bz{{\bf z}}
 \global\long\def\bh{{\bf h}}

\global\long\def\bw{{\bf w}}

\global\long\def\bba{{\bf A}}
 \global\long\def\bbb{{\bf B}}

\global\long\def\bbf{{\bf F}}
 \global\long\def\bbx{{\bf X}}
 \global\long\def\bbu{{\bf U}}
 \global\long\def\bbi{{\bf I}}
 \global\long\def\argmin{{\displaystyle \mathop{\mbox{argmin}}}}
 \global\long\def\argmax{{\displaystyle \mathop{\mbox{argmax}}}}

\global\long\def\supp{{\displaystyle \mathop{\mbox{{\rm supp}}}}}

\title{GESPAR: Efficient Phase Retrieval of\\
 Sparse Signals}

\author{Yoav Shechtman, Amir Beck, and Yonina~C.~Eldar,~\IEEEmembership{Fellow,~IEEE}
\thanks{The work of A. Beck is supported in part by the Israel Science Foundation
under Grant no. 253/12. The work of Y. Eldar is supported in part
by the Israel Science Foundation under Grant no. 170/10, in part by
the Ollendorf Foundation, in part by the SRC, and in part by the Intel
Collaborative Research Institute for Computational Intelligence (ICRI-CI).%
} %
\thanks{Y. Shechtman is with the department of Physics, Technion--Israel Institute
of Technology, Haifa, Israel 32000 (e-mail: joe@tx.technion.ac.il).%
} %
\thanks{A. Beck is with the department of Industrial Engineering, Technion--Israel
Institute of Technology, Haifa, Israel 32000 (e-mail: becka@ie.technion.ac.il).%
} %
\thanks{Y. C. Eldar is with the department of Electrical Engineering, Technion--Israel
Institute of Technology, Haifa, Israel 32000 (e-mail: yonina@ee.technion.ac.il
). %
}}
\maketitle
\begin{abstract}
We consider the problem of phase retrieval, namely,
recovery of a signal from the magnitude of its Fourier transform,
or of any other linear transform. Due to the loss of the Fourier phase
information, this problem is ill-posed. Therefore, prior information
on the signal is needed in order to enable its recovery. In this work
we consider the case in which the signal is known to be sparse, i.e.,
it consists of a small number of nonzero elements in an appropriate
basis. We propose a fast local search method for recovering a sparse
signal from measurements of its Fourier transform (or other linear transform) magnitude which
we refer to as GESPAR: GrEedy Sparse PhAse Retrieval. Our algorithm
does not require matrix lifting, unlike previous approaches, and therefore
is potentially suitable for large scale problems such as images. Simulation
results indicate that GESPAR is fast and more accurate than existing
techniques in a variety of settings.
\end{abstract}

\section{Introduction}

Recovery of a signal from the magnitude of its Fourier transform,
also known as phase retrieval, is of great interest in applications
such as optical imaging \cite{optics}, crystallography \cite{crystallography},
and more \cite{PR-book}. Due to the loss of Fourier phase information,
the problem (in 1D) is generally ill-posed. A common approach to overcome
this ill-posedeness is to exploit prior information on the signal.
A variety of methods have been developed that use such prior information,
which may be the signal's support (region in which the signal is nonzero),
non-negativity, or the signal's magnitude \cite{Fienup}, \cite{GS}.

A popular class of algorithms is based on the use of alternate projections
between the different constraints. In order to increase the probability
of correct recovery, these methods require the prior information to
be very precise, for example, exact/or ``almost'' exact knowledge
of the support set. Since the projections are generally onto non-convex
sets, convergence to a correct recovery is not guaranteed \cite{convex}.
A more recent approach is to use matrix-lifting of the problem which
allows to recast phase retrieval as a semi-definite programming (SDP)
problem \cite{candesPR}. The algorithm developed in \cite{candesPR}
does not require prior information about the signal but instead uses
multiple signal measurements (e.g., using different illumination settings,
in an optical setup).

In order to obtain more robust recovery without requiring multiple
measurements, we develop a method that exploits signal sparsity. Existing
approaches aimed at recovering sparse signals from their Fourier magnitude
belong to two main categories: SDP-based techniques \cite{Yoav},\cite{Babak},\cite{Henrik},\cite{Mallat}
and algorithms that use alternate projections (Fienup-type methods)
\cite{sparseFienup}. Phase retrieval of sparse signals can be viewed
as a special case of the more general quadratic compressed sensing
(QCS) problem considered in \cite{Yoav}. Specifically, QCS treats
recovery of sparse vectors from quadratic measurements of the form
$y_{i}=\mathbf{x}^{T}\mathbf{A}_{i}\mathbf{x},\,\,\, i=1,\ldots,N$,
where $\mathbf{x}$ is the unknown sparse vector to be recovered,
$y_{i}$ are the measurements, and $\mathbf{A}_{i}$ are known matrices.
In (discrete) phase retrieval, $\mathbf{A}_{i}=\bbf_{i}^{*}\bbf_{i}$
where $\bbf_{i}$ is the $i$th row of the discrete Fourier transform
(DFT) matrix. QCS is encountered, for example, when imaging a sparse
object using partially spatially-incoherent illumination \cite{Yoav}.

A general approach to QCS was developed in \cite{Yoav} based on matrix
lifting. More specifically, the quadratic constraints where lifted
to a higher dimension by defining a matrix variable $\bbx=\bx\bx^{T}$.
The problem was then recast as an SDP involving minimization of the
rank of the lifted matrix subject to the recovery constraints as well
as row sparsity constraints on $\bbx$. An iterative thresholding
algorithm based on a sequence of SDPs was then proposed to recover
a sparse solution. Similar SDP-type ideas were recently used in the
context of phase retrieval \cite{Babak},\cite{Henrik}. However,
due to the increase in dimension created by the matrix lifting procedure,
the SDP approach is not suitable for large-scale problems.

Another approach for phase retrieval of sparse signals is adding a
sparsity constraint to the well-known iterative error reduction algorithm
of Fienup \cite{sparseFienup}. In general, Fienup-type approaches
are known to suffer from convergence issues and often do not lead
to correct recovery especially in 1D problems; simulation results
show that even with the additional information that the input is sparse,
convergence is still problematic and the algorithm often recovers
erroneous solutions.

In this paper we propose an efficient method for phase retrieval which
also leads to good recovery performance. Our approach is based on
a fast 2-opt local search method (see \cite{P98} for an excellent
introduction to such techniques) applied to a sparsity constrained
non-linear optimization formulation of the problem. We refer to the
resulting algorithm as GESPAR: GrEedy Sparse PhAse Retrieval. Sparsity
constrained nonlinear optimization problems have been considered recently
in \cite{AmirYonina}; the method derived in this paper is motivated
-- although different in many aspects -- by the local search-type
techniques of \cite{AmirYonina}. In essence, GESPAR is a local-search
method, where the support of the sought signal is updated iteratively,
according to selection rules described in detail in Section~\ref{sec:GrEedy-Sparse-PhAse}.
A local minimum of the objective function is then found given the
current support using the damped Gauss Newton algorithm. Theorem~\ref{thm:dgn}
establishes convergence of the iterations to a stationary point of
the objective under suitable conditions.

We demonstrate through numerical simulations that GESPAR is both efficient
and more accurate than current techniques. Several other aspects of
the algorithm are explored via simulations such as robustness to noise,
and scalability for larger dimensions. In the simulations performed
we found that the number of measurements needed for reliable recovery
from Fourier magnitudes seems to scale like $s^{3}$, where $s$ is
the sparsity level.

GESPAR is applicable to recovery of a sparse vector from general quadratic
measurements, and is not restricted to Fourier magnitude measurements.
Nonetheless, when the measurements are obtained in the Fourier domain,
the algorithm can be implemented efficiently by exploiting the fast
Fourier transform, as we discuss in Section~\ref{sec:f_implement}.

The remainder of the paper is organized as follows. We formulate the
problem in Section~\ref{sec:Problem-Formulation}. Section~\ref{sec:GrEedy-Sparse-PhAse}
describes our proposed algorithm in detail and establishes convergence
of the local iterations. Implementation details for Fourier-based
problems are provided in Section~\ref{sec:f_implement}. Extensive
numerical experiments illustrating the empirical performance of GESPAR
are presented in Section~\ref{sec:Numerical-Simulations}.

\section{\label{sec:Problem-Formulation}Problem Formulation}

\subsection{Sparse Phase Retrieval: Fourier Measurements}

We are given a vector of measurements $\by\in\mathbb{{R}}^{N}$, that
corresponds to the magnitude-squared of an $N$ point DFT of a vector
$\bx\in\mathbb{{R}}^{N}$, i.e.:
\begin{equation}
y_{l}=\left|\sum_{m=1}^{n}x_{m}e^{-\frac{2\pi j(m-1)(l-1)}{N}}\right|^{2},\,\,\, l=1,\ldots,N.
\end{equation}
Here $\bx$ is constructed by $(N-n)$ zero padding of a vector
$\bar{\bx}\in\real^{n}$ with elements $x_{i},\; i=1,2,\ldots,n$.
Denoting by $\bbf\in\mathbb{{C}}^{N\times N}$ the DFT matrix with
elements $\exp\big\{-\frac{2\pi j(m-1)(l-1)}{N}\big\}$, we can express
$\by$ as $\by=|\bbf\bx|^{2}$, where $|\cdot|^{2}$ denotes the element-wise
absolute-squared value. The vector $\bar{\bx}$ is known to be $s$-sparse,
that is, it contains at most $s$ nonzero elements. Our goal is to
recover $\bar{\bx}$, or $\bx$, given the measurements $\by$ and
the sparsity level $s$.

The mathematical formulation of the problem that we consider consists
of minimizing the sum of squared errors subject to the sparsity constraint:
\begin{equation}
\begin{array}{ll}
\min_{\bx} & \sum_{i=1}^{N}(|\bbf_{i}\bx|^{2}-y_{i})^{2}\\
\mbox{s.t.} & \|\bx\|_{0}\leq s,\\
 & \text{supp}(\bx)\subseteq\{1,2,\ldots,n\},\\
 & \bx\in\mathbb{R}^{N},
\end{array}\label{eq:2}
\end{equation}
 where $\bbf_{i}$ is the $i$th row of the DFT matrix $\bbf$, and
$\|\cdot\|_{0}$ stands for the zero-``norm'', that is, the number
of nonzero elements. Note that the unknown vector $\bx$ can only
be found up to trivial degeneracies that are the result of the loss
of Fourier phase information: circular shift, global phase, and signal
``mirroring''.

\subsubsection*{Support Information}

To aid in solving the phase retrieval problem, we can rely on the
fact that the autocorrelation sequence of $\bar{\bx}$ (the first
$n$ components of $\bx$) may be determined from $\by$ if $N\geq2n-1$. Specifically,
let
\begin{equation}
g_{m}=\sum_{i=1}^{n}x_{i}x_{i+m},\quad m=-(n-1),\ldots,n-1
\end{equation}
 denote the correlation sequence of length $2n-1$. If we choose $N\geq2n-1$,
then $\{g_{m}\}$ can be obtained by taking the inverse DFT of $\by$.

Determining $g_{m}$ requires oversampling, or
zero-padding of $\mathbf{x}$. While this additional information improves
the recovery performance, as is demonstrated in the simulations section,
it is not actually needed for GESPAR to work. Nevertheless, when this
information is available, GESPAR exploits it, in the following way.
First of all, we assume that no support cancelations occur in $\{g_{m}\}$,
namely, if $x_{i}\neq0$ and $x_{j}\neq0$ for some $i,j$, then $g_{|i-j|}\neq0$.
When the values of $\bx$ are random, this is true with probability
1. This fact can be used in GESPAR in order to obtain initial information
on the support of $\bx$, which we capture by two sets $J_1$ and $J_2$.

Denote by $J_{1}$ the set of indices known
in advance to be in the support.
To derive the set $J_{1}$, note that due to the existing degree of
freedom relating to shift-invariance of $\bx$, the index 1 can be
assumed to be in the support, thereby removing this degree of freedom;
as a consequence, the index corresponding to the last nonzero element
in the autocorrelation sequence is also in the support, i.e.
\[
i_{max}=1+\underset{i}{\text{argmax}}\{i:g_{i}\neq0\}.
\]
 Therefore, $J_{1}=\{1,i_{max}\}.$

Next, we denote by $J_{2}$ the set of indices that are candidates
for being in the support, meaning the indices that are \textit{not}
known in advance to be in the off-support (the complement of the support).
Specifically, $J_{2}$ contains the set of all indices $k\in\{1,2,\ldots,n\}$
such that $g_{k-1}\neq0$. Obviously, since we assume that $x_{k}=0$
for $k>n$, we have $J_{2}\subseteq\{1,2,\ldots,n\}$. As a concrete
example, consider the signal $\mathbf{\bar{x}}=(2,0,0,-1,0,-1.5)^{T}$.
The corresponding $11$ point autocorrelation function $g_{m}$ is
given by $g_{m}=(-3,0,-2,1.5,0,7.25,0,1.5,-2,0,-3)^{T}$. The set
$J_{1}$ is therefore $J_{1}=\{1,6\}$. Next, by examining the zeros
of $g_{m}$, and using our assumption of no support-cancelations,
we deduce that there are no two non-zero elements $x_{i}\neq0$ and
$x_{j}\neq0$ such that $|i-j|=1,4$. Therefore, forcing the first
element in $\mathbf{x}$ to be non-zero, which removes the shift-invariance
degeneracy, immediately implies that $x_{2}=x_{5}=0.$ In this way
$J_{2}$ is determined as $J_{2}=\{1,3,4,6\}.$ Defining $\bba_{i}=\Re(\bbf_{i})^{T}\Re(\bbf_{i})+\Im(\bbf_{i})^{T}\Im(\bbf_{i})\in\real^{N\times N}$,
problem (\ref{eq:2}) along with the support information can be written
as
\begin{equation}
\begin{array}{ll}
\min_{\bx} & f(\bx)\equiv\sum_{i=1}^{N}(\bx^{T}\bba_{i}\bx-y_{i})^{2}\\
\mbox{s.t.} & \|\bx\|_{0}\leq s,\\
 & J_{1}\subseteq\supp(\bx)\subseteq J_{2},\\
 & \bx\in\mathbb{R}^{N},
\end{array}\label{eq:2reform}
\end{equation}
 which will be the formulation to be studied.

Note that even with knowledge of the exact support of $\mathbf{x}$
there is no guarantee for uniqueness beyond the aforementioned trivial
degeneracies. Consider for example the two vectors $\mathbf{u}=(1,0,-2,0,-2)$
and $\mathbf{v=}(1-\sqrt{3},0,1,0,1+\sqrt{3})$. Both of these vectors
are $s=3$ sparse, and they have the same autocorrelation function
$g_{m}=(-2,0,2,0,9,0,2,0,-2)$. This ambiguity therefore cannot be
resolved using any method that uses sparsity (even exact support information)
and autocorrelation (or Fourier magnitude) measurements alone.

Finally, when the measurements are noisy, the autocorrelation information
is not very useful for support estimation, since very small (noise
level) values in the autocorrelation sequence cannot be treated as
zero. For this reason, the autocorrelation-derived support information
is not used in GESPAR at all in the noisy case. Formally, ignoring
this information is equivalent to setting $J_{1}=\{1\}$ and $J_{2}=\{1,2,\ldots,n\}$.

\subsection{Sparse Phase Retrieval: General Measurements}

Although the problem formulation above assumes Fourier measurements
and sparsity of $\mathbf{\bar{x}}$, we show below that our approach
applies to arbitrary quadratic measurements of \textbf{$\mathbf{\bar{x}}$}.
This includes the case in which $\bar{\mathbf{x}}$ is sparse in a basis other than the identity basis.
In fact, in this general case, the formulation given in \eqref{eq:2reform}
remains the same, with the only change being the definition of the
matrices $\mathbf{A}_{i}$.

Consider the phase retrieval problem with respect to arbitrary linear measurements, so that
\begin{equation}
\mathbf{y}_{i}=|\inner{\bm{\phi}_{i}}{\mathbf{x}}|^{2},
\end{equation}
for a set of measurement vectors $\bm{\phi}_{i}\in \mathbb{R}^n,i=1,\dots,N$.
The corresponding phase retrieval problem can be written as in
\eqref{eq:2reform} with $\mathbf{A}_{i}=\bm{\phi}_{i}\bm{\phi}_{i}^{T}$.
Similarly, suppose that \textbf{$\mathbf{\bar{x}}=\mathbf{D}\mathbf{z}$}, where
$\mathbf{D\in\mathbb{R}^{\mathnormal{n\times b}}}$ is some basis in
which $\mathbf{\bar{x}}$ is sparse, and $\mathbf{z}\in\mathbb{R}^{b}$
is a sparse vector. In this case
$\mathbf{A}_{i}=\mathbf{D}^{T}\bm{\phi}_{i}\bm{\phi}_{i}^{T}\mathbf{D}$.
Thus, our formulation can accommodate arbitrary sparsity bases and general quadratic measurements.

In the next section, we propose GESPAR---an iterative local-search
based algorithm for solving (\ref{eq:2reform}). We note that although
in the context of phase retrieval the parameters $\bba_{i},J_{1},J_{2}$
have special properties (e.g., $\bba_{i}$ is positive semidefinite
of at most rank 2, $|J_{1}|=2$), we will not use these properties
in GESPAR. Therefore, our approach is capable of handling general
instances of (\ref{eq:2reform}) with the sole assumption that $\bba_{i}$
is symmetric for any $i=1,2,\ldots,N$. In the Fourier case, the algorithm
can be implemented more efficiently, as we discuss in Section~\ref{sec:f_implement}.

\section{\label{sec:GrEedy-Sparse-PhAse}GrEedy Sparse PhAse Retrieval (GESPAR) }

\subsection{The Damped Gauss-Newton Method}

Before describing our algorithm, we begin by presenting a variant
of the damped Gauss-Newton (DGN) method \cite{B99},\cite{B96} that
is in fact the core step of our approach. The DGN method is invoked
in order to solve the problem of minimizing the objective function
$f$ over a \textit{given} support $S\subseteq\{1,2,\ldots,n\}\;(|S|=s)$:
\begin{equation}
\min\{f(\bbu_{S}\bz):\bz\in\mathbb{R}^{s}\},\label{fus}
\end{equation}
 where $\bbu_{S}\in\mathbb{R}^{N\times s}$ is the matrix consisting
of the columns of the identity matrix $\bbi_{N}$ corresponding to
the index set $S$. With this notation, (\ref{fus}) can be explicitly
written as
\begin{equation}
\min\left\{ g(\bz)\equiv\sum_{i=1}^{N}(\bz^{T}\bbu_{S}^{T}\bba_{i}\bbu_{S}\bz-y_{i})^{2}:\bz\in\real^{s}\right\} .\label{uau}
\end{equation}

The minimization in (\ref{uau}) is a nonlinear least-squares problem.
A natural approach for tackling it is via the DGN method. This algorithm
begins with an arbitrary vector $\bz_{0}$. In our simulations, we
choose it as a white random Gaussian vector with zero mean and unit
variance. At each iteration, all the terms inside the squares in $g(\bz)$
are linearized around the previous guess. Namely, we write $g(\bz)$
from (\ref{uau}) as:
\begin{equation}
g(\mathbf{z})=\sum_{i=1}^{N}h_{i}^{2}(\mathbf{z}),
\end{equation}
 with $h_{i}(\bz)=\bz^{T}\bbb_{i}\bz-y_{i}$, and $\bbb_{i}=\bbu_{S}^{T}\bba_{i}\bbu_{S}$.
At each step we replace $h_{i}$ by its linear approximation around
$\bz_{k-1}$:
\begin{eqnarray}
h_{i} & \approx & h_{i}(\bz_{k-1})+\nabla h_{i}(\bz_{k-1})^{T}(\bz-\bz_{k-1})\nonumber \\
 & = & \bz_{k-1}^{T}\bbb_{i}\bz_{k-1}-y_{i}+2(\bbb_{i}\bz_{k-1})^{T}(\bz-\bz_{k-1}).
\end{eqnarray}
 We then choose $\bz_{k}$ to be the solution of the problem {\small
\begin{equation}
\underset{\mathbf{z}}{\text{min}}\sum_{i=1}^{N}(\bz_{k-1}^{T}\bbb_{i}\bz_{k-1}-y_{i}+2(\bbb_{i}\bz_{k-1})^{T}(\bz-\bz_{k-1}))^{2}.\label{eq:taylor}
\end{equation}
 }{\small \par}

Problem (\ref{eq:taylor}) can be written as a linear least-squares
problem
\begin{equation}
\tilde{\mathbf{z}}_{k}=\arg\min\|J(\bz_{k-1})\bz-\mathbf{b}_{k}\|_{2}^{2}\, \label{eq:ls_problem}
\end{equation}
 with the $i$th row of $J(\bz_{k-1})$ being $\nabla h_{i}(\bz_{k-1})^{T}=2(\bbb_{i}\bz_{k-1})^{T}$,
and the $i$th component of $\bb_{k}$ given by $y_{i}+\bz_{k-1}^{T}\bbb_{i}\bz_{k-1}$
for $i=1,2,\ldots,N$. The solution $\tilde{\bz}_{k}$ is equal to
$\tilde{\bz}_{k}=(J(\bz_{k-1})^{T}J(\bz_{k-1}))^{-1}J(\bz_{k-1})^{T}\bb_{k}$.
We then define a direction vector $\bd_{k}=\bz_{k-1}-\tilde{\bz}_{k}$.
This direction is used to update the solution with an appropriate
stepsize designed to guarantee the convergence of the method to a
stationary point of $g(\bz)$. The stepsize is chosen via a simple
backtracking procedure. Algorithm \ref{alg:dgn} describes the DGN
method in detail. In our implementation the stopping parameters were
chosen as $\varepsilon=10^{-4}$ and $L=100$.

The following theorem establishes the rate of convergence of the norm
of the gradient of the objective function to zero, and consequently
proves that the limit points of the sequence are stationary points.

\begin{theorem} \label{thm:dgn} Let $\{\bz_{k}\}$ be the sequence
generated by the DGN method. Assume that $\sum_{i=1}^{N}\bbb_{i}\succ\bo$
and that there exists $\underline{\lambda}>0$ such that for all $k$
\[
\lambda_{\min}(J(\bz_{k})^{T}J(\bz_{k}))\geq\underline{\lambda}.
\]
 Then $\nabla g(\bz_{k})\rightarrow\bo$ as $k\rightarrow\infty$
and there exists a constant $C>0$ such that
\begin{equation}
\min_{p=1,\ldots,k}\|\nabla g(\bz_{p})\|\leq\frac{\sqrt{g(\bz_{0})}}{C\sqrt{k+1}}.\label{mr}
\end{equation}
 Moreover, each limit point of the sequence is a stationary point
of $g$. \end{theorem}

\begin{proof} See Appendix~\ref{sec:proof}. \end{proof}

Note that the proof requires $J(\bz_{k})$ to have full column rank,
and in fact that the minimum eigenvalues of $J(\bz_{k})^{T}J(\bz_{k})$
are uniformly bounded below. In the vast majority of our runs this
assumption held true; however, we did encounter in our numerical experiments
a few cases in which this condition was not valid. In these situations,
our implementation chose one of the optimal solutions of the corresponding
least-squares problem. We noticed that these cases had negligible
effect on the results.

\noindent
\begin{algorithm}[ftp]
\caption{\label{alg:dgn}DGN for solving (\ref{uau})}

\textbf{Input:} $(\bba_{i},y_{i},S,\varepsilon,L)$. \\
 $\bba_{i}\in\real^{N\times N},i=1,2,\ldots,N$ - symmetric matrices.\\
 $y_{i}\in\real,i=1,2,\ldots,N.$\\
 $S\subseteq\{1,2,\ldots,n\}$ - index set.\\
 $\varepsilon$ - stopping criteria parameter.\\
 $L$ - maximum allowed iterations.

\textbf{Output:} $\bz$ - an optimal (or suboptimal) solution of (\ref{uau}).
\\
 \line(1,0){250}

\textbf{Initialization:} Set $\bbb_{i}=\bbu_{S}^{T}\bba_{i}\bbu_{S},t_{0}=0.5$,
$\bz_{0}$ a random vector. \\

\textbf{General Step $k(k\geq1)$:} Given the iterate $\bz_{k-1}$,
the next iterate is determined as follows:

1. \textbf{Gauss-Newton Direction:} Let $\tilde{\bz}_{k}$ be the
solution of the linear least-squares problem (\ref{eq:ls_problem}), given
by:
\[
\tilde{\bz}_{k}=(J(\bz_{k-1})^{T}J(\bz_{k-1}))^{-1}J(\bz_{k-1})^{T}\bb_{k}
\]
 with the $i$th row of $J(\bz_{k-1})$ being $2(\bbb_{i}\bz_{k-1})^{T}$,
and the $i$th component of $\bb_{k}$ given by $y_{i}+\bz_{k-1}^{T}\bbb_{i}\bz_{k-1}$. The Gauss-Newton direction is
\[
\bd_{k}=\bz_{k-1}-\tilde{\bz}_{k}.
\]
 \\
 2. \textbf{Stepsize Selection via Backtracking:} set $u=\min\{2t_{k-1},1\}$.
Choose a stepsize $t_{k}$ as $t_{k}=(\frac{1}{2})^{m}u$, where $m$
is the minimal nonnegative integer for which
\[
g\left(\bz_{k-1}-\left(\frac{1}{2}\right)^{m}u\bd_{k}\right)<g(\bz_{k-1})-u\left(\frac{1}{2}\right)^{m+1}\nabla g(\bz_{k-1})^{T}\bd_{k},
\]
 with $g(\bz)$ given by \eqref{uau}. \\
 3. \textbf{Update:} set $\bz_{k}=\bz_{k-1}-t_{k}\bd_{k}.$\\
 4. \textbf{Stopping rule:} STOP if either $\|\bz_{k}-\bz_{k-1}\|<\varepsilon$
or $k>L$.
\end{algorithm}

\subsection{The 2-opt Local Search Method}

The GESPAR method consists of repeatedly invoking a local-search method
on an initial random support set. In this section we describe the
local search procedure. At the beginning, the support is chosen to
be a set of $s$ random indices chosen to satisfy the support constraints
$J_{1}\subseteq S\subseteq J_{2}$. Then, at each iteration a swap
between a support and an off-support index is performed such that
the resulting solution via the DGN method improves the objective function.
Since at each iteration only two elements are changed (one in the
support and one in the off-support), this is a so-called ``2-opt''
method (see \cite{P98}). The swaps are always chosen to be between
the index corresponding to components in the current iterate $\bx_{k-1}$
with the smallest absolute value and the off-support index corresponding
to the component of $\nabla f(\bx_{k-1})=4\sum_{i}(\mathbf{x}_{k-1}^{T}\mathbf{A}_{i}\mathbf{x}_{k-1}-\mathbf{c}_{i})\mathbf{A}_{i}\mathbf{x}_{k-1}$
with the largest absolute value. This process continues as long as
the objective function decreases and stops when no improvement can
be made. A detailed description of the method is given in Algorithm
\ref{alg:2opt}.

\begin{algorithm}
\caption{\label{alg:2opt}2-opt}

\textbf{Input}: ($\bba_{i},y_{i}$).\\
 $\bba_{i}\in\real^{N\times N},i=1,2,\ldots,N$ - symmetric matrices.\\
 $y_{i}\in\real,i=1,2,\ldots,N.$\\

\textbf{Output}: $\bx$ - a suggested solution for problem (\ref{eq:2reform}).\\
 $T$ - total number of required swaps. \\
 \line(1,0){250}
\begin{enumerate}
\item \textbf{Initialization:}

\begin{enumerate}
\item Set $T=0$.
\item Generate a random index set $S_{0}(|S_{0}|=s)$ satisfying the support
constraints ($J_{1}\subseteq S_{0}\subseteq J_{2}$).
\item Invoke the DGN method with parameters $(\bba_{i},y_{i},S_{0},10^{-4},100)$
and obtain an output $\bz_{0}$. Set $\bx_{0}=\bbu_{S_{0}}\bz_{0}$.
\end{enumerate}
\item \textbf{General Step ($k=1,2,\ldots$):}

\begin{enumerate}
\item Let $i$ be the index from $S_{k-1}\backslash J_{1}$ corresponding
to the component of $\bx_{k-1}$ with the smallest absolute value.
Let $j$ be the index from $S_{k-1}^{c}\cap J_{2}$ corresponding
to the component of $\nabla f(\bx_{k-1})$ with the highest absolute
value.
\item Set $\tilde{S}=S_{k-1}$, and make a swap between the indices $i$
and $j$
\[
\tilde{S}=(S_{k-1}\backslash\{i\})\cup\{j\}.
\]
 Invoke DGN with input $(\bba_{i},y_{i},\tilde{S},10^{-4},100)$ and
obtain an output $\tilde{\bz}$. Set $\tilde{\bx}=\bbu_{S}\tilde{\bz}$.
Advance $T$: $T\leftarrow T+1$. \\
 If $f(\tilde{\bx})<f(\bx_{k-1})$, then set $S_{k}=\tilde{S},\bx_{k}=\tilde{\bx}$,
advance $k$ and goto 2.a.
\item If none of the swaps resulted with a better objective function value,
then STOP. The output is $\bx=\bx_{k-1}$ and $T$. \end{enumerate}
\end{enumerate}
\end{algorithm}

\subsection{The GESPAR Algorithm}

The 2-opt method can have the tendency to get stuck at local optima
points. Therefore, our final algorithm, which we call GESPAR, is a
restarted version of 2-opt. The 2-opt method is repeatedly invoked
with different initial random support sets until the resulting objective
function value is smaller than a certain threshold (success) or the
number of maximum allowed total number of swaps was passed (failure).
A detailed description of the method is given in Algorithm \ref{gespar}.
One element of our specific implementation that is not described in
Algorithm \ref{gespar} is the incorporation of random weights added
to the objective function, giving randomly different weights to the
different measurements. Namely, the objective function used is actually
chosen as $f(\bx)=\sum_{i=1}^{N}w_{i}(\bx^{T}\bba_{i}\bx-y_{i})^{2}$
with \textbf{$w_{i}=1$} or \textbf{$2$ }with equal probability.
The random generation of weights is done each time the DGN procedure
is invoked. We observed that this modification reduced the probability
of the 2-opt procedure to get stuck in non-optimal points.

\begin{algorithm}
\caption{\label{gespar}GESPAR}

\textbf{Input:} ($\bba_{i},y_{i},\tau,{\rm ITER}$).\\
 $\bba_{i}\in\real^{N\times N},i=1,2,\ldots,N$ - symmetric matrices.\\
 $y_{i}\in\real,i=1,2,\ldots,N.$\\
 $\tau$ - threshold parameter. \\
 ITER - Maximum allowed total number of swaps. \\

\textbf{Output:} $\bx$ - an optimal (or suboptimal) solution of (\ref{eq:2reform}).

\line(1,0){250}

\textbf{Initialization}. Set $C=0,k=0$. \\

\begin{itemize}
\item \textbf{Repeat}\\
 Invoke the 2-opt method with input ($\bba_{i},y_{i}$) and obtain
an output $\bx$ and $T$. Set $\bx_{k}=\bx,C=C+T$ and advance $k$:
$k\leftarrow k+1$.\\
 \textbf{Until} $f(\bx)<\tau$ or $C>{\rm ITER}$.
\item The output is $\bx_{\ell}$ where ${\ell}=\underset{m=0,1...k-1}{\text{argmin}}f(\bx_{m})$. \end{itemize}
\end{algorithm}

\section{Fourier Implementation Details}

\label{sec:f_implement} In principle, GESPAR may be used to find
sparse solutions to any system of quadratic equations, i.e. problems
of the form:
\begin{equation}
\begin{array}{ll}
\min_{\bx} & \sum_{i=1}^{N}(\bx^{T}\bba_{i}\bx-y_{i})^{2}\\
\mbox{s.t.} & \|\bx\|_{0}\leq s,\\
 & \bx\in\mathbb{R}^{N}.
\end{array}
\end{equation}
 However, when the matrices $\bba_{i}$ correspond to transforms that
can be implemented efficiently, GESPAR takes on a particularly simple
form.

For example, consider the case in which $\{\bba_{i}\}$ represent
Fourier measurements. In this case, the creation and storing of the
matrices $\bba_{i}$ defined in Section~\ref{sec:Problem-Formulation},
can be avoided in the implementation, by using the FFT. Specifically,
to calculate the weighted objective function, we note that
\begin{equation}
f(\bx)=\sum_{i=1}^{N}w_{i}(\bx^{T}\bba_{i}\bx-y_{i})^{2}=\sum_{i=1}^{N}w_{i}(|\hat{x}_{i}|^{2}-y_{i})^{2}
\end{equation}
 where $\hat{x}_{i}$ is the $i$th DFT component of $\bx$, which
can be computed via the FFT. Clearly, $J(\bz)$, which is used in the DGN procedure (Algorithm \ref{alg:dgn}) can also be computed
efficiently since $\bbb_{i}=\bbu_{S}^{T}\bba_{i}\bbu_{S}$ only involves
a small ($s$) number of columns of the Fourier matrix $\mathbf{F}$.

The FFT can also be used in the calculation of the gradient $\nabla f(\bx)$,
used in the 2-opt stage \ref{alg:2opt}:
\begin{align}
\nabla f(\bx) & =4\sum_{i}w_{i}(\mathbf{x}^{T}\mathbf{A}_{i}\mathbf{x}-y_{i})\mathbf{A}_{i}\mathbf{x}\nonumber \\& =
4N \mbox{IFFT}[(|\hat{x}_i|^2-y_i)w_i\hat{x}_i].
\end{align}
Consequently, in no step of the algorithm is it necessary to calculate
the set of matrices $\mathbf{A}_{i}$ explicitly.

This fact is even more important in the 2D Fourier phase
retrieval problem, as the relevant vector sizes become very large.
Since a major advantage of GESPAR over other methods (e.g. SDP based)
is its low computational cost, GESPAR may be used to find a sparse
solution to the 2D Fourier phase retrieval - or phase retrieval of
images. The only adjustments needed in the algorithm are in the implementation,
for example, using FFT2 instead of storing the large matrices $\bba_{i}$.

Figure \ref{fig:2D-Fourier} shows a recovery example of a sparse
$195\times195$ pixel image, comprised of $s=15$ circles at random
locations and random values on a grid containing $225$ points, recovered
from its $38,025$ 2D-Fourier magnitude measurements, using GESPAR.
The dictionary used in this example contains 225 elements consisting
of non-overlapping circles located on a $15\times15$ point cartesian
grid, each with a 13 pixel diameter. The solution took 80 seconds.
Solving the same problem using the sparse Fienup algorithm did not
yield a successful reconstruction, and using the SDP method is not
practical due to the large matrix sizes.

Further investigation of the algorithm's performance in the 2D case
is presented in Section~\ref{sec:Numerical-Simulations}.

\begin{figure*}
\begin{centering}
\includegraphics[width=0.95\textwidth]{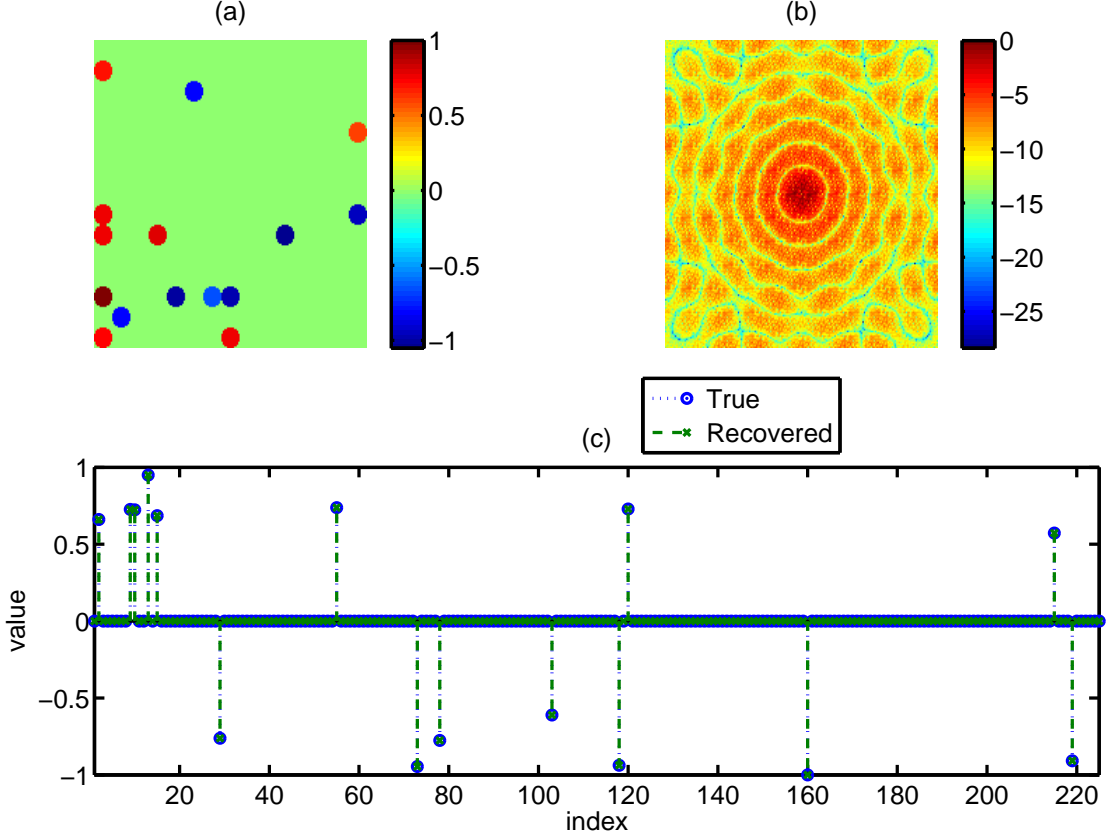}
\par\end{centering}

\caption{\label{fig:2D-Fourier}2D Fourier phase retrieval example. (a) True
$195\times195$ sparse circle image ($s=15$ circles). (b) Measured
2D Fourier magnitude ($38,025$ measurements, log scale). (c) True
and recovered coefficient vectors, corresponding to circle amplitudes
at each of the $225$ grid points.}
\end{figure*}

\section{\label{sec:Numerical-Simulations}Numerical Simulations}

In order to demonstrate the performance of GESPAR, we conduct several
numerical simulations. The algorithm is compared to other existing
methods, and is evaluated in terms of signal-recovery accuracy, computational
efficiency, and robustness to noise.

\subsection{Signal-recovery Accuracy\label{sub:Signal-recovery-Accuracy}}

In this subsection we examine the recovery success rate of GESPAR
as a function of the number of non-zero elements in the signal. A
runtime comparison of the tested methods is also performed.

We choose $\bar{\bx}$ as a random vector of length $n$. The vector
contains uniformly distributed values in the range $[-4,-3]\cup[3,4]$
in $s$ randomly chosen elements. The $N$ point DFT of the signal
is calculated, and its magnitude-square is taken as $\by$, the vector
of measurements. The $2n-1$ point correlation is also calculated.
In order to recover the unknown vector $\bx$, the GESPAR algorithm
is used with $\tau=10^{-4}$ and $ITER=6400$. We also test two other
algorithms for comparison purposes: An SDP based algorithm (Algorithm
2, \cite{Babak}), and an iterative Fienup algorithm with a sparsity
constraint \cite{sparseFienup}. In our simulation $n=64$ and $N=128$.
The Sparse-Fienup algorithm is run using $100$ random initial points,
out of which the chosen solution is the one that best matches the
measurements. Namely, $\hat{\mathbf{x}}$ is selected as the $s$
sparse output of the Sparse-Fienup algorithm with the minimal cost
$f(\bx)=\sum_{i=1}^{N}(|\bbf_{i}\bx|^{2}-y_{i})^{2}$ out of the $100$
runs.

Signal recovery results of the numerical simulation are shown in Fig.~\ref{fig:recovery},
where the probability of successful recovery is plotted for different
sparsity levels. The success probability is defined as the ratio of
correctly recovered signals $\bx$ out of $100$ simulations. In each
simulation both the support and the signal values are randomly selected.
The three algorithms (GESPAR, SDP and Sparse-Fienup) are compared.
The results clearly show that GESPAR outperforms the other methods
in terms of probability of successful recovery - over 90\% successful
recovery up to $s=15$, vs. $s=8$ and $s=7$ in the other two techniques.

Average runtime comparison of the three algorithms is shown in Table~\ref{table:runtime}
for $n=64$ and $N=128$. The runtime is averaged over all successful
recoveries. The computer used has an intel i5 CPU and 4GB of RAM.
As seen in the table, the SDP based algorithm is significantly slower
than the other two methods. Fienup iterations are slightly slower
than GESPAR and lead to a much lower success rate. In these simulations,
GESPAR is both fast and more accurate than its competitors.

\begin{figure}
\begin{centering}
\centering \includegraphics[width=0.9\columnwidth]{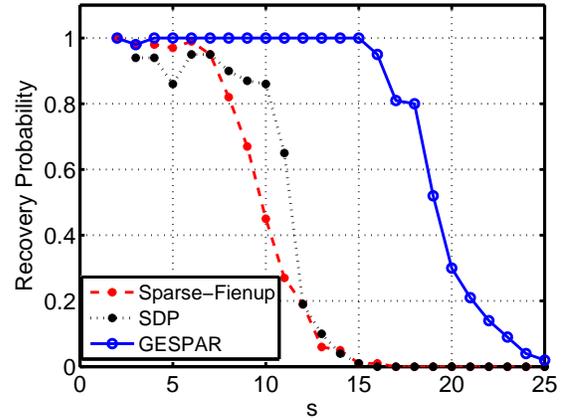}
\par\end{centering}

\caption{Recovery probability vs. sparsity (s)}

\label{fig:recovery}
\end{figure}

\begin{table}
\centering \caption{Runtime comparison}

\label{table:runtime}

\begin{tabular}{|>{\centering}p{0.08\linewidth}|>{\centering}p{0.09\columnwidth}|>{\centering}p{0.09\columnwidth}|c|>{\centering}p{0.09\columnwidth}|>{\centering}p{0.09\linewidth}|>{\centering}p{0.09\columnwidth}|}
\hline
 & \multicolumn{2}{c|}{SDP} & \multicolumn{2}{c|}{Sparse-Fienup} & \multicolumn{2}{c|}{GESPAR}\tabularnewline
\hline
\hline
 & recovery \%  & runtime sec  & recovery \%  & runtime sec  & recovery \%  & runtime sec\tabularnewline
\hline
\hline
$s=3$  & 0.93  & 1.32  & 0.98  & 0.09  & 1  & 0.04 \tabularnewline
\hline
\hline
$s=5$  & 0.86  & 1.78  & 0.97  & 0.12  & 1  & 0.05 \tabularnewline
\hline
\hline
$s=8$  & 0.9  & 3.85  & 0.82  & 0.50  & 1  & 0.06 \tabularnewline
\hline
\end{tabular}
\end{table}

\subsection{Sensitivity to exact sparsity knowledge}

Since the exact value of the signal's sparsity $s$ may not be known,
the performance of GESPAR is examined when only an upper limit on
$s$ is given. To this end we run GESPAR twice: Once with $s$ known
exactly at each realization, and once with only an upper limit on
$s$. The upper limit is taken as $25$. The other simulation settings
are the same as in Section\ref{sub:Signal-recovery-Accuracy}.

Figure \ref{fig:Effect-of-unknown} shows the probability for successful
recovery of the two simulations. The rather loose upper limit on $s$
does not seem to affect the results significantly--- in fact, the
performance is somewhat improved when allowing more nonzero elements
during the iterations.

\begin{figure}
\begin{centering}
\centering \includegraphics[width=0.9\columnwidth]{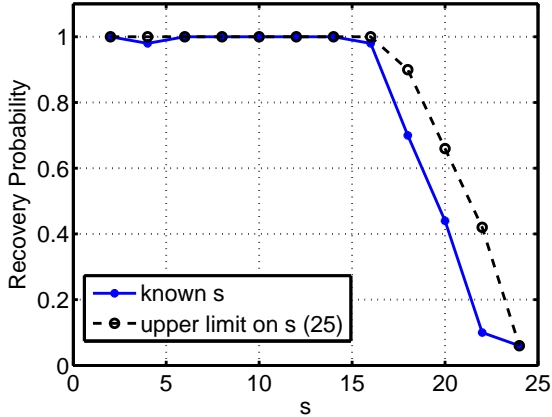}
\par\end{centering}

\caption{\label{fig:Effect-of-unknown}Effect of unknown exact sparsity level
$s$ on recovery probability. }
\end{figure}

\subsection{Effect of the number of allowed swaps}

One of the stopping criteria for the GESPAR algorithm is when the
total number of swaps exceeds a predefined parameter (the input parameter
$ITER$ in Algorithm \ref{gespar}). Naturally, increasing the allowed
number of index swaps will increase the probability of finding a correct
solution, but at the cost of increased computation time. It is therefore
important to quantify this effect, which is the purpose of the current
simulation.

We run GESPAR with the same parameters as in Section\ref{sub:Signal-recovery-Accuracy}
several times, where in each simulation a different value for the
parameter $ITER$ is used, in the range $[100,25600]$. Figure \ref{fig:ITER}
shows the results. As expected, increasing the number of possible
swaps increases the recovery probability. Note that increasing the
value of $ITER$ above $6400$ demonstrated no improvement in the
recovery results - for the unsuccessful recoveries, increasing the
number of swaps even up to $ITER=25600$ did not help. This means
that for these simulation values (e.g. $N=128,\, s<25$), using a
value of $ITER$ larger than $6400$ only increases computation time
without improving the results.

\begin{figure}
\begin{centering}
\includegraphics[width=0.9\columnwidth]{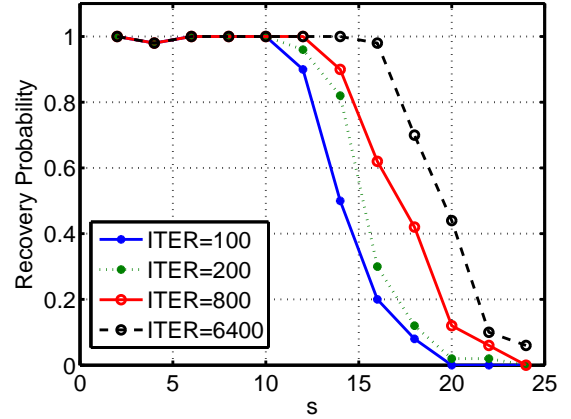}
\par\end{centering}

\centering{}\caption{\label{fig:ITER}Effect of number of swaps ($ITER$) on recovery probability}
\end{figure}

\subsection{Effect of oversampling and support information}

Here we examine the effect of oversampling and of autocorrelation-derived
support information. GESPAR is run on random vectors $\mathbf{x}$
of length $n=64$, with a varying amount of noiseless Fourier magnitude
measurements, obtained by the $N$ point DFT of \textbf{$\mathbf{x}$
}with $N=64,128,256$. In these cases, no support information was
used - i.e. $J_{1}=\{1\}$ and $J_{2}=\{1,2,\ldots,n\}$. In addition,
in order to investigate the effect of support information, we run
GESPAR with $n=64,N=128$ (i.e. oversampling by a factor of 2), and
use the support information derived from the autocorrelation sequence.
The results, shown in Fig. \ref{fig:Effect-of-oversampling}, clearly
show that both oversampling and support information improve performance.

\begin{figure}
\begin{centering}
\includegraphics[width=0.9\columnwidth]{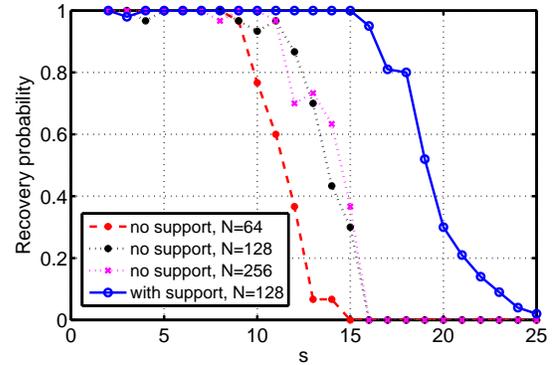}
\par\end{centering}

\caption{\label{fig:Effect-of-oversampling}Effect of oversampling and support
information from the autocorrelation sequence ($n=64$).}

\end{figure}

\subsection{Robustness to noise}

We now evaluate GESPAR as a function of SNR, and compare it with sparse
Fienup \cite{sparseFienup}. The SDP based method presented in \cite{Babak}
is not designed to deal with noise and therefore we did not apply
it here. The SDP approach of \cite{Henrik} considers random measurements,
and does not produce comparable results from direct Fourier measurements.

As in Section~\ref{sub:Signal-recovery-Accuracy}, we choose $\mathbf{\bar{x}}$
as a vector of length $n$, with $s$ randomly chosen elements containing
uniformly distributed values, and evaluate its $N$ point Fourier
magnitude-square. White-gaussian noise $\mathbf{v}$ is added to the
measurements, at different SNR values, defined as: $SNR=20\text{log}\frac{\|\mathbf{y\|}}{\|\mathbf{v\|}}$.
In order to recover the unknown vector $\bx$, the GESPAR algorithm
is used with $\tau=10^{-4}$ and $ITER=10000$, as well as the sparse-Fienup
algorithm, for comparison purposes. In our simulation $n=64$ and
$N=128$. The sparse-Fienup algorithm is run with a maximum of $1000$
iterations, and with $100$ random initial points.

Note that even with little noise, the information on the support obtained
by the zeros of the autocorrelation is no longer available. This is
due to the fact that in the presence of noise, there will be no true
zeros in the measured (or calculated) autocorrelation. In this case,
one might try to threshold the autocorrelation values, rendering small
autocorrelation values as zeros. However, this might result in zeroing
of small (yet non-zero) values of the true autocorrelation function.
Therefore, in the noisy case, we do not use support information obtained
by the autocorrelation function in GESPAR, namely $J_{2}=\{1,2,\ldots,n\}$.

Figure \ref{fig:Normalized-MSE-vs.} shows the normalized mean squared
reconstruction error (NMSE), defined as $NMSE=\frac{\|\mathbf{x}-\mathbf{\hat{x}}\|_{2}}{\|\mathbf{x}\|_{2}}$,
as a function of sparsity, for different SNR values. Each point represents
an average over 100 different random realizations. The performance
under different SNR values is plotted for GESPAR (full lines), and
for sparse-Fienup (dashed-lines). The performance of GESPAR naturally
improves as SNR increases, and it clearly outperforms sparse-Fienup
in terms of noise-robustness.

\begin{figure*}
\centering \includegraphics[width=0.9\textwidth]{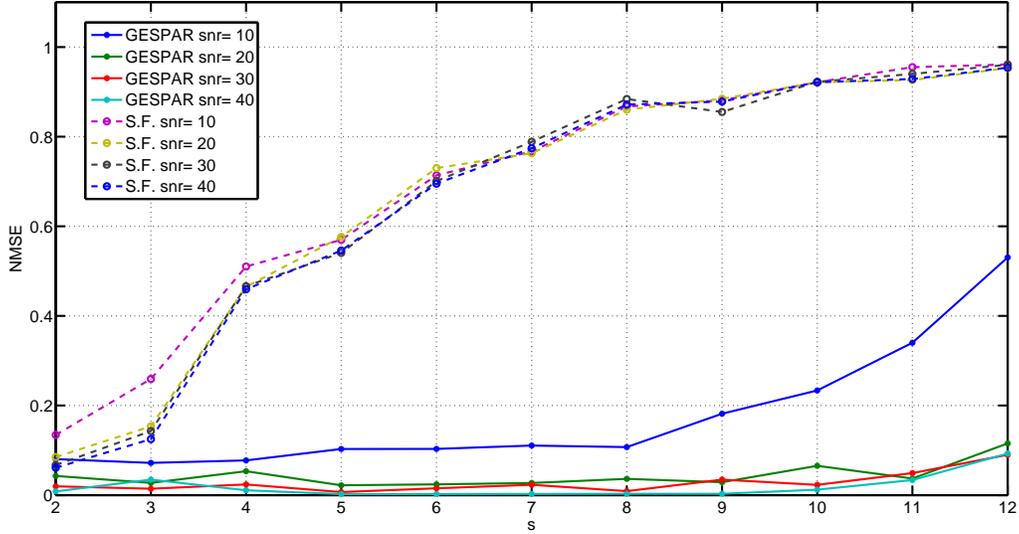}

\caption{\label{fig:Normalized-MSE-vs.}Normalized MSE vs. sparsity level.
The performance is plotted for several SNR values for GESPAR (full
lines) and Sparse-Fienup (S.F. - dashed lines).}
\end{figure*}

\subsection{\label{sub:Scalability}Scalability}

As one of the main advantages of GESPAR over SDP based methods is
its ability to solve large problems efficiently, we now examine its
performance for different vector sizes.

We simulate GESPAR for various values of $n\in[64,2048]$. In all
cases $N=2n$. The other simulation parameters are as in Section~\ref{sub:Signal-recovery-Accuracy}.
The recovery probability vs. sparsity $s$ for different vector lengths
is shown in Fig.~\ref{fig:Scalability}. The maximal sparsity $s$
allowing successful recovery is shown to increase with vector length
$n$, and seems to scale like $n^{1/3}$, which is consistent with
the same scaling observation presented in \cite{Babak}. The mean
reconstruction time for a signal with $n=512,\, s=35$ from $N=1024$
measurements, allowing $ITER=6400$ replacements, is $33.5$ seconds.
For comparison, a corresponding plot representing the scalability
of the sparse-Fienup algorithm is presented in Fig.~\ref{fig:Sparse-Fienup-scalability}.
Plotting a similar scalability plot for the SDP based method is not
possible due to the high computational cost which under our simulation
conditions limits the application of this method to around $n\sim400$.

\begin{figure}
\begin{centering}
\includegraphics[width=0.9\columnwidth]{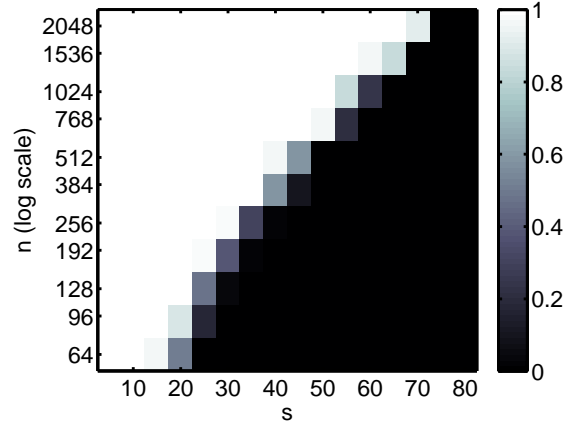}
\par\end{centering}

\caption{\label{fig:Scalability}1D-Scalability - GESPAR recovery probability
as a function of signal sparsity $s$, for various vector lengths
($n\in[64,2048]$), and with oversampling, i.e. $N=2n$. White corresponds
to high recovery probability.}
\end{figure}

\begin{figure}
\begin{centering}
\includegraphics[width=0.9\columnwidth]{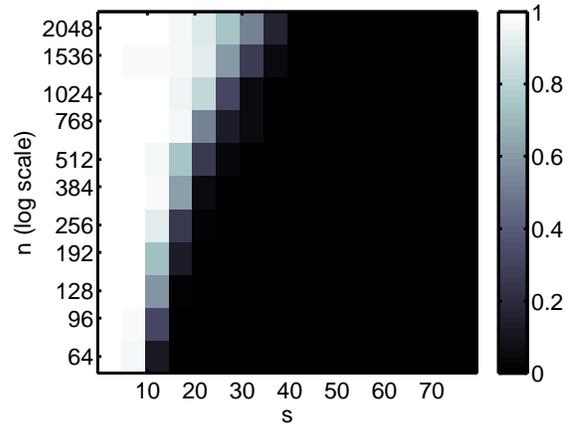}
\par\end{centering}

\caption{\label{fig:Sparse-Fienup-scalability}Sparse Fienup scalability -
recovery probability as a function of signal sparsity $s$, for various
vector lengths ($n\in[64,2048]$), and with oversampling, i.e. $N=2n$.
White corresponds to high recovery probability. }
\end{figure}

\subsection{Computation Time}

The most time consuming part of GESPAR is the matrix inversion process
in the DGN segment of the algorithm. Therefore, computation time scales
approximately linearly with the number of swaps - as each swap corresponds
to a single DGN iteration. The approximately linear dependence of
runtime in the number of swaps is displayed in Fig.~\ref{fig:Time-vs.-number}.
Each point in the plot represents the mean time it took GESPAR to
run $ITER$ iterations, averaged over 50 random input signals with
$N=128,\, n=64,s=10$.

A major factor that determines the computation time is the number
of index swaps required to find a solution. The mean number of swaps
as a function of $s,n$ is shown in Fig.~\ref{fig:Number-of-swaps}.
Beyond the successful recovery region (the white region in Fig.~\ref{fig:Scalability}),
the maximal number of swaps ($6400$) is used, without yielding a
correct solution.

\begin{figure}
\begin{centering}
\includegraphics[width=0.9\columnwidth]{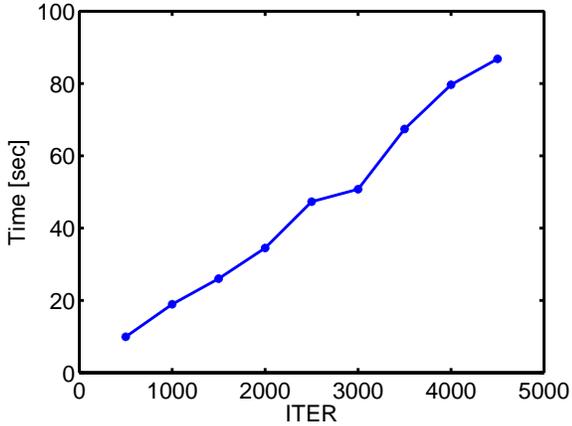}
\par\end{centering}

\caption{\label{fig:Time-vs.-number}Time vs. number of swaps ($ITER$).}
\end{figure}

\begin{figure}
\begin{centering}
\includegraphics[width=0.9\columnwidth]{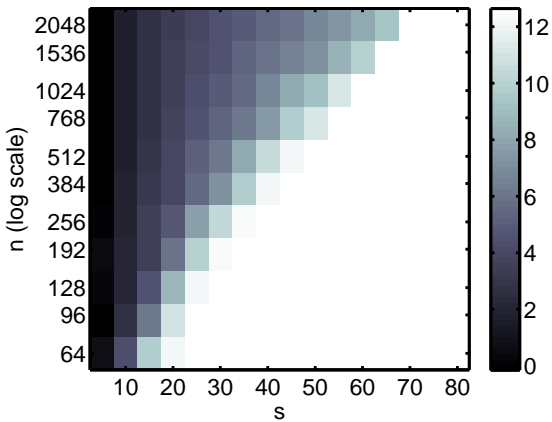}
\par\end{centering}

\caption{\label{fig:Number-of-swaps}Number of swaps as a function of $s$
and $n$. The colorbar is in $log_{2}$ scale, e.g. $10\Rightarrow2^{10}=1024$
swaps. }
\end{figure}

\subsection{Two-Dimensional Fourier Phase Retrieval}

In this section we apply GESPAR to 2D Fourier phase retrieval problems,
showing its ability to solve large scale problems efficiently.

We generate random $s-$sparse 2D signals of sizes $\sqrt{n}\times\sqrt{n}$,
with varying values for $s$ and $n$, in the ranges $s\in[2:82]$
and $n\in[256:6400]$. Each signal is recovered from the noiseless
magnitude of its 2D DFT, with no oversampling, using GESPAR. Similarly
to the 1D noisy simulation, no autocorrelation-derived support information
was used here. The parameter $ITER$ is taken as $6400$. The recovery
probability vs. sparsity $s$ for different vector lengths is shown
in Fig.~\ref{fig:Scalability-in-2D}. Similarly to the 1D case, the
maximal sparsity allowing successful recovery increases with $n$.
For comparison, Fig.~\ref{fig:Sparse-Fienup-2D-Scalability--} shows
the result of a sparse-Fienup scalability simulation for the 2D case,
under the same conditions, with 200 initial points per signal (increasing
this parameter did not affect the results significantly). GESPAR is
shown to outperform the sparse-Fienup method in the 2D case as well.\textbf{
}As in the 1D case, a comparison to SDP based methods is not included
here, since applying the SDP based method on the 2D case is very difficult
due to memory limitations.

A comparison between GESPAR and the sparse-Fienup method is shown
in Fig.~\ref{fig:Runtime-comparison--}. The comparison shows the
average time a successful recovery in the simulation took, as a function
of vector size $n$. Sparse-Fienup is seen to be faster, however
comparing Fig.~\ref{fig:Scalability-in-2D} to Fig.~\ref{fig:Sparse-Fienup-2D-Scalability--}
shows that GESPAR can recover signals up with a higher value of $s$:
For example, when $n=6400$, GESPAR recovers with very high probability
signals up to sparsity $s=57$, while sparse Fienup only recovers up to
$s=42$.

\begin{figure}
\begin{centering}
\includegraphics[width=0.9\columnwidth]{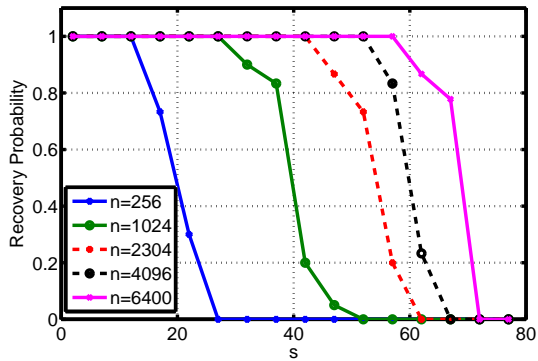}
\par\end{centering}

\caption{\label{fig:Scalability-in-2D}GESPAR 2D-Scalability - recovery probability
as a function of signal sparsity for various image sizes ($n=256,1024,2304,4096,6400$).}
\end{figure}

\begin{figure}
\centering{}\includegraphics[width=0.9\columnwidth]{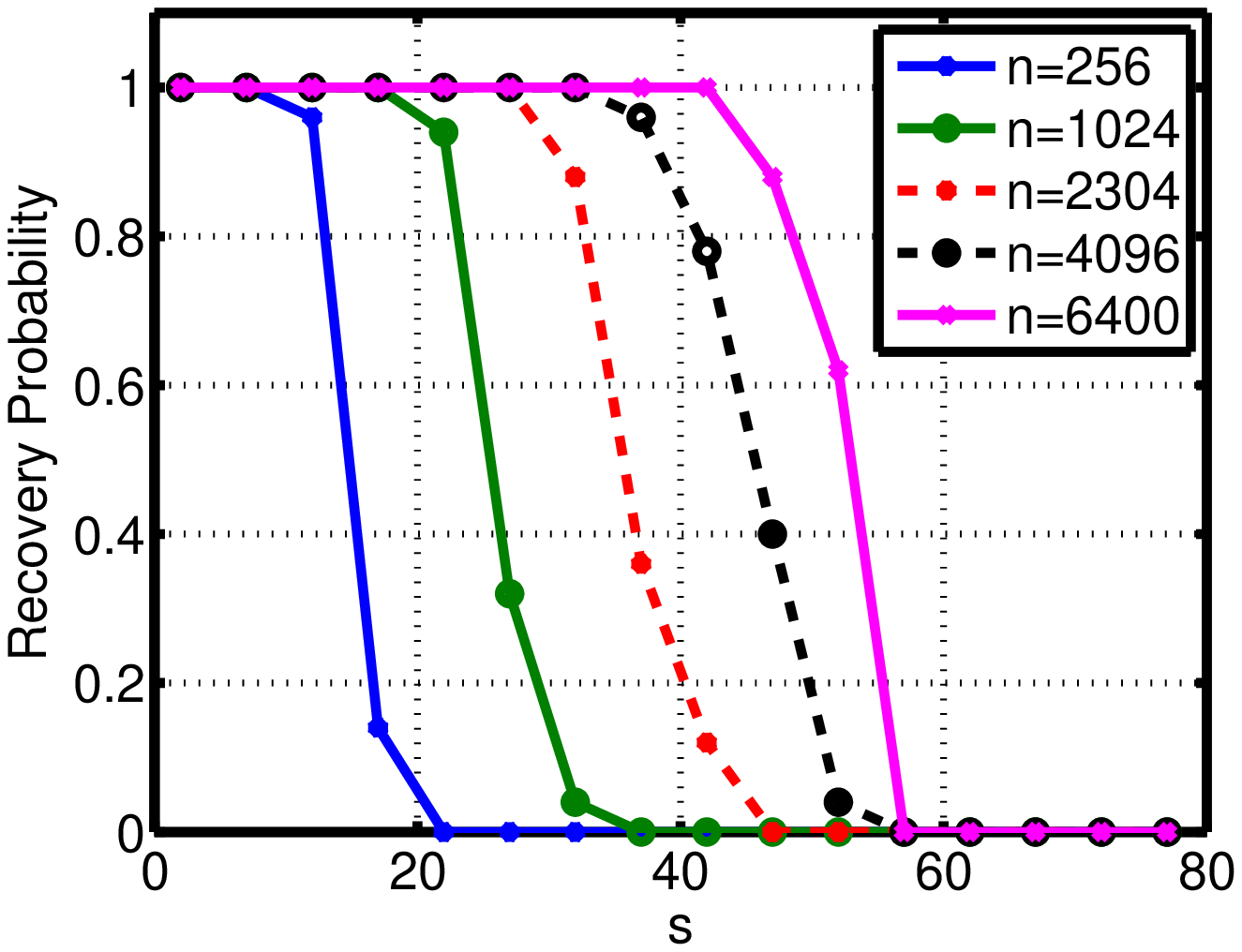}\caption{\label{fig:Sparse-Fienup-2D-Scalability--}Sparse-Fienup 2D-Scalability
- recovery probability as a function of signal sparsity for various
image sizes ($n=256,1024,2304,4096,6400$).}
\end{figure}

\begin{figure}
\begin{centering}
\includegraphics[width=0.9\columnwidth]{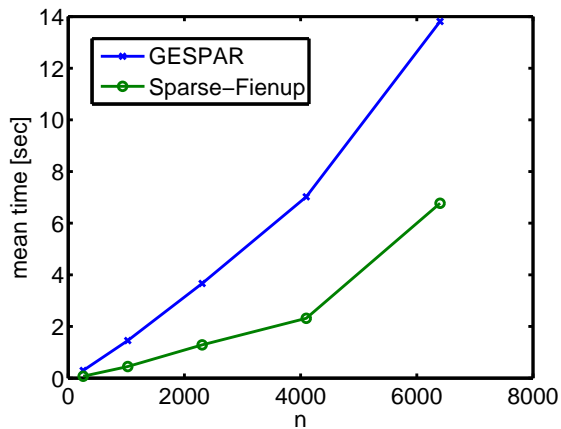}
\par\end{centering}

\caption{\label{fig:Runtime-comparison--}Runtime comparison - average computation
time for a succesful 2D recovery, for GESPAR and for sparse-Fienup,
as a function of $n$.}

\end{figure}

\section{Conclusion}

We proposed and demonstrated GESPAR - a fast algorithm for recovering
a sparse vector from its Fourier magnitude, or more generally, from
quadratic measurements. We showed via simulations that GESPAR outperforms
alternative approaches suggested for this problem in terms of complexity
and success probability. The algorithm does not require matrix-lifting,
and therefore is potentially suitable for large scale problems such
as 2D images. The simulations demonstrated robustness of GESPAR to
noise and other inexact knowledge, as well as its ability to successfully
treat a variety of phase retrieval problems in one and two dimensions.

\bibliography{IEEEabrv,GESPAR_arxiv_bib}

\appendix

\section{Proof of Theorem}

\label{sec:proof} Define the vector-valued function $\bh$ by
\[
\bh(\bz)=(h_{1}(\bz),h_{2}(\bz),\ldots,h_{N}(\bz))^{T},
\]
 with $h_{i}(\bz)=\bz^{T}\bbb_{i}\bz-y_{i}$ With this notation, the
vector $\bb_{k}$ can be written as
\[
\bb_{k}=J(\bz_{k-1})\bz_{k-1}-\bh(\bz_{k-1}),
\]
 and the solution of the least-squares problem is
\begin{eqnarray}
\tilde{\bz}_{k} & = & (J(\bz_{k-1})^{T}J(\bz_{k-1}))^{-1}J(\bz_{k-1})^{T}\nonumber \\
 &  & (J(\bz_{k-1})\bz_{k-1}-\bh(\bz_{k-1}))\nonumber \\
 & = & \bz_{k-1}-(J(\bz_{k-1})^{T}J(\bz_{k-1}))^{-1}J(\bz_{k-1})^{T}\bh(\bz_{k-1})\nonumber \\
 & = & \bz_{k-1}-\frac{1}{2}(J(\bz_{k-1})^{T}J(\bz_{k-1}))^{-1}\nabla g(\bz_{k-1}).\label{eq:des}
\end{eqnarray}
 Finally,
\begin{equation}
\bd_{k}=\frac{1}{2}(J(\bz_{k-1})^{T}J(\bz_{k-1}))^{-1}\nabla g(\bz_{k-1}).\label{def:dk}
\end{equation}
 From (\ref{eq:des}) it follows that $-\bd_{k}$ is a descent direction
since
\begin{eqnarray}
\lefteqn{-\bd_{k}^{T}\nabla g(\bz_{k-1})}\nonumber \\
 & = & -\frac{1}{2}\nabla g(\bz_{k-1})(J(\bz_{k-1})^{T}J(\bz_{k-1}))^{-1}\nabla g(\bz_{k-1})<0.
\end{eqnarray}

We now show that the sequence generated by the DGN method is bounded.
Indeed, since $-\bd_{k}$ is a descent direction,
\begin{eqnarray*}
\sqrt{g(\bz_{0})} & \geq & \sqrt{g(\bz_{k})}\\
 & = & \sqrt{{\textstyle \sum_{i=1}^{N}(\bz_{k}^{T}\bbb_{i}\bz_{k}-y_{i})^{2}}}\\
 & \geq & \frac{1}{\sqrt{N}}{\textstyle \sum_{i=1}^{N}|\bz_{k}^{T}\bbb_{i}\bz_{k}-y_{i}|}\\
 & \geq & \frac{1}{\sqrt{N}}\left(\bz_{k}^{T}({\textstyle \sum_{i=1}^{N}\bbb_{i})\bz_{k}-\sum_{i=1}^{N}y_{i}}\right),
\end{eqnarray*}
 where the second inequality is due to Cauchy-Schwarz and the last
inequality is a result of the fact that $\sum_{i=1}^{N}\bbb_{i}\succ0$
and $y_{i}\geq0$. Therefore,
\[
\|\bz_{k}\|^{2}\leq\frac{1}{\lambda_{\min}(\sum_{i=1}^{N}\bbb_{i})}\left(\sqrt{Ng(\bz_{0})}+{\textstyle \sum_{i=1}^{N}y_{i}}\right)\equiv\alpha,
\]
 proving that $\{\bz_{k}\}\subseteq B[\bo,\sqrt{\alpha}]=\{\bz:\|\bz\|\leq\sqrt{\alpha}\}.$

Since $g$ is twice continuously differentiable, and $J(\bz)$ is
continuous, it follows that there exists $M>0$ and $\Lambda>0$ such
that $\lambda_{\max}(\nabla^{2}g(\bz))\leq M$ and $\lambda_{\max}(J(\bz)^{T}J(\bz))\leq\Lambda$
for any $\bz\in B[\bo,2\sqrt{\alpha}]$. In addition, since $\nabla g$
is continuous over $B[\bo,2\sqrt{\alpha}]$, there exist $\beta>0$
such that $\|\nabla g(\bz)\|\leq\beta$ for all $\bz\in B[0,2\sqrt{\alpha}]$.
Therefore, by (\ref{def:dk}) it follows that
\begin{equation}
\|\bd_{k}\|\leq\frac{\beta}{2\underline{\lambda}}.\label{bound:dk}
\end{equation}
 The fact that $\lambda_{\max}(\nabla^{2}g(\bz))\leq M$ for all $\bz\in B[\bo,2\sqrt{\alpha}]$
implies that $\nabla g$ is Lipschitz continuous over $B[\bo,2\sqrt{\alpha}]$
with parameter $M>0$. Hence, by the descent lemma \cite{B99},
\begin{equation}
g(\by)\leq g(\bx)+\nabla g(\bx)^{T}(\by-\bx)+\frac{M}{2}\|\by-\bx\|^{2}\label{descent_lemma}
\end{equation}
 for any $\bx,\by\in B[\bo,2\sqrt{\alpha}]$.

From $\|\bz_{k-1}\|\leq\sqrt{\alpha}$ and $\|\bd_{k}\|\leq\beta/(2\underline{\lambda})$,
it follows that $\bz_{k-1}-t\bd_{k}\in B[\bo,2\sqrt{\alpha}]$ whenever
$t\leq\frac{2\underline{\lambda}\sqrt{\alpha}}{\beta}$. Therefore,
we can plug $\by=\bz_{k-1}-t\bd_{k}$ and $\bx=\bz_{k-1}$ into (\ref{descent_lemma})
to obtain
\[
g(\bz_{k-1}-t\bd_{k})\leq g(\bz_{k-1})-t\nabla g(\bz_{k-1})^{T}\bd_{k}+\frac{Mt^{2}}{2}\|\bd_{k}\|^{2}.
\]
 Using \eqref{def:dk},
\begin{eqnarray*}
\|\bd_{k}\|^{2} & = & \frac{1}{4}\nabla g(\bx)^{T}(J(\bz_{k})^{T}J(\bz_{k}))^{-2}\nabla g(\bx)\\
 & \leq & \frac{1}{4\underline{\lambda}}\nabla g(\bx)^{T}(J(\bz_{k})^{T}J(\bz_{k}))^{-1}\nabla g(\bx)\\
 & = & \frac{1}{2\underline{\lambda}}\nabla g(\bz_{k-1})^{T}\bd_{k},
\end{eqnarray*}
 which yields
\[
g(\bz_{k-1})-g(\bz_{k-1}-t\bd_{k})\geq t\left(1-\frac{M}{4\underline{\lambda}}t\right)\nabla g(\bz_{k-1})^{T}\bd_{k}.
\]
 Therefore, if $t\leq\min\left\{ \frac{2\underline{\lambda}}{M},\frac{2\underline{\lambda}\sqrt{\alpha}}{\beta}\right\} $,
then
\begin{equation}
g(\bz_{k-1})-g(\bz_{k-1}-t\bd_{k})\geq\frac{t}{2}\nabla g(\bz_{k-1})^{T}\bd_{k}.\label{gzk}
\end{equation}

By the way the backtracking procedure is defined, we have that either
$t_{k}=1$ or $2t_{k}>\min\left\{ \frac{2\underline{\lambda}}{M},\frac{2\underline{\lambda}\sqrt{\alpha}}{\beta}\right\} $
and hence $t_{k}\geq\min\left\{ 1,\frac{\underline{\lambda}}{M},\frac{\underline{\lambda}\sqrt{\alpha}}{\beta}\right\} $.
Together with (\ref{gzk}) this results in the inequality
\begin{eqnarray*}
g(\bz_{k-1})-g(\bz_{k}) & \geq & \frac{t_{k}}{2}\nabla g(\bz_{k-1})^{T}\bd_{k}\\
 & \geq & \min\left\{ \frac{1}{2},\frac{\underline{\lambda}}{2M},\frac{\underline{\lambda}\sqrt{\alpha}}{2\beta}\right\} \nabla g(\bz_{k-1})^{T}\bd_{k}.
\end{eqnarray*}
 Since
\begin{eqnarray*}
\lefteqn{\nabla g(\bz_{k-1})^{T}\bd_{k}}\\
 & = & \nabla g(\bx_{k-1})^{T}(J(\bz_{k-1})^{T}J(\bz_{k-1}))^{-1}\nabla g(\bz_{k-1})\\
 & \geq & \frac{1}{\Lambda}\|\nabla g(\bz_{k-1})\|^{2},
\end{eqnarray*}
 we conclude that
\begin{equation}
g(\bz_{k-1})-g(\bz_{k})\geq C\|\nabla g(\bz_{k-1})\|^{2},\label{gdif}
\end{equation}
 where $C=\min\left\{ \frac{1}{2\Lambda},\frac{\underline{\lambda}}{2M\Lambda},\frac{\underline{\lambda}\sqrt{\alpha}}{2\beta\Lambda}\right\} $.
Noting that $\{g(\bz_{k})\}$ is a bounded below and nonincreasing
sequence, it follows that it converges. The left-hand side of (\ref{gdif})
therefore converges to zero and we obtain the result that $\nabla g(\bz_{k})$
converges to zero as $k$ tends to infinity. This fact also readily
implies that all accumulation points of the sequence are stationary.
Summing the inequality (\ref{gdif}) over $p=1,2,\ldots,k+1$ we obtain
that
\[
g(\bz_{0})-g(\bz_{k+1})\geq C{\textstyle \sum_{p=1}^{k+1}\|\nabla g(\bz_{p-1})\|^{2},}
\]
 and consequently, (also using the fact that $g(\bz_{k+1})\geq0$),
\[
g(\bz_{0})\geq C{\displaystyle (k+1)\min_{p=1,\ldots,k+1}\|\nabla g(\bz_{p-1})\|^{2},}
\]
 from which the inequality (\ref{mr}) follows. $\Box$
\end{document}